\documentclass[10pt,twocolumn,twoside] {IEEEtran} 

\usepackage{amsmath,amssymb,graphicx,epsfig,fancyhdr,amsthm,tabulary}
\usepackage{url,hyperref}
\usepackage{cite}
\usepackage{multirow}
\usepackage{float}
\usepackage{amsfonts}
\usepackage{footnote}
\makesavenoteenv{tabular}
\usepackage{tabulary,bigstrut,array,multirow,booktabs}%
\usepackage{tikz}
\usetikzlibrary{calc}
\usepackage{pdfsync}
\usepackage{bm}
\usepackage{latexsym}
\usepackage{flushend}

\newtheorem{thm}{Theorem}[]

\newtheorem{cor}{Corollary}
\newtheorem{lem}{Lemma}
\newtheorem{pro}{Proposition}
\newtheorem{property}{Property}

\theoremstyle{remark}
\newtheorem{rem}{Remark}

\theoremstyle{definition}
\newtheorem{defn}{Definition}

\newcommand{\tr}{\mathop{\mathrm{tr}}}

\newcommand{\llangle}{\langle \! \langle} 
\newcommand{\rrangle}{\rangle \! \rangle}

\begin{document}

\title{Systematic DFT Frames: Principle, Eigenvalues Structure, and Applications}

\author{\IEEEauthorblockN{Mojtaba Vaezi, {\it Student Member, IEEE,}  and Fabrice Labeau, {\it Senior Member, IEEE}}\\
}


\maketitle

\begin{abstract}
Motivated by a host of recent applications requiring some amount of redundancy,
{\it frames} are becoming a standard tool in the signal processing toolbox. In this paper,
we study a specific class of frames, known as discrete Fourier transform (DFT) codes,
and introduce the notion of {\it systematic}
frames for this class. This is encouraged by a new application of frames, namely,
distributed source coding that uses DFT codes for compression.
Studying their extreme eigenvalues, we show that, unlike DFT frames, systematic
DFT frames are not necessarily {\it tight}. Then, we come up with conditions
for which these frames can be tight. 
In either case, the best and worst systematic frames are established in the minimum mean-squared reconstruction error sense.
Eigenvalues of DFT frames and their subframes play a pivotal role in this work.
Particularly, we derive some bounds on the extreme eigenvalues DFT subframes
which are used to prove most of the results; these bounds are valuable independently.

\end{abstract}


\IEEEpeerreviewmaketitle

\section{Introduction}
{\let\thefootnote\relax\footnotetext{
Copyright (c) 2013 IEEE. Personal use of this material is permitted.
However, permission to use this material for any other purposes must
be obtained from the IEEE by sending a request to pubs-permissions@ieee.org.

This work was supported by Hydro-Qu\'{e}bec, the Natural Sciences and Engineering
Research Council of Canada and McGill University in the framework of the
NSERC/Hydro-Qu\'{e}bec/McGill Industrial Research Chair in Interactive
Information Infrastructure for the Power Grid.
Part of the material in this paper was presented at the International Symposium on Information
Theory, Boston, MA, July 2012 \cite{Vaezi2012frame}. 	

The authors are with the Department of Electrical and Computer
Engineering, McGill University, Montreal, QC H3A 0E9, Canada (e-mail: 	
mojtaba.vaezi@mail.mcgill.ca; fabrice.labeau@mcgill.ca). 	
}}

\IEEEPARstart{F}{rames}, ``redundant" set of vectors used for signal representation,
 are increasingly found in signal processing applications.
Frames are more general than bases as frames are complete but not necessarily linearly independent.
A {\it basis}, on the contrary, is a set of vectors  used
to ``uniquely" represent a vector as a linear combination of basis elements.
Frames are generally motivated by applications requiring
some level of {\it redundancy}, and
they offer flexibility in design, resilience to additive noise (including quantization error),
robustness to erasure (loss), and numerical stability of reconstruction. 
With increasing applications, frames are becoming more prevalent in signal processing.

 In this paper, we study a specific
class of frames known as discrete Fourier transform (DFT) codes.
By using these codes, the ideas of coding theory are
described in the signal processing setting. We consider
the Bose-Chaudhuri-Hocquenghem (BCH) codes, an important class of multiple-error-correcting codes,
in the DFT domain \cite{blahut2003algebraic,Marshall,wolf1983redundancy}.
BCH-DFT codes are {\it cyclic codes} in the complex (or real) domain,
similar to BCH codes in the binary error correction
setting.
Therefore, their codewords have certain successive spectral
components equal to zero. This property is then exploited for error detection
and correction in the complex (real) field
\cite{Marshall,marvasti1999efficient,blahut2003algebraic,wolf1983redundancy,rath2004subspace,rath2004subspace1,gabay2007joint,Takos}.

From the frame theory perspective, DFT codes are harmonic tight frames.
In the absence of erasure, tight frames minimize the mean-squared error (MSE) between
the transmitted and received signals \cite{goyal2001quantized, rath2004frame, kovacevic2007life}.
The MSE is the ultimate measure of performance in many digital communication systems
where quantized analog signal is transmitted.
Frames are naturally robust to transmission loss since they provide an overcomplete expansion of signal
\cite{goyal2001quantized,casazza2003equal, kovacevic2007life, rath2004frame, kovacevic2008introduction}.

DFT frames have recently been applied in the context of {\it distributed source coding} (DSC) \cite{vaezi2011DSC}.
More precisely, BCH-DFT codes are used for compression of analog signals with side information
available at the decoder.
In DSC context, side information is viewed as corrupted version of signal, and
compression is achieved by sending only redundant information,
in the form of {\it parity} or {\it syndrome}, with respect to a channel code \cite{dragotti2009distributed}.
Unlike in DSC that uses binary channel codes for compression,
in the new framework (DSC based on BCH-DFT codes) compression is performed before quantization.
As a result, DFT frames, which are primarily used for compression, can decrease quantization error at the same time.
This results in a better reconstruction, in the MSE sense, particularly
when the sources are highly correlated.

Motivated by its application in parity-based DSC \cite{vaezi2011DSC}
and {\it distributed joint source-channel coding} (DJSCC) \cite{vaezi2012WZ}
that use DFT codes, we introduce the notion of {\it systematic frames},
in this work.
For an $(n,k)$ frame, a systematic frame is defined to be a frame that includes the identity matrix of size $k$ as a subframe.
Since {\it tight} frames minimize reconstruction error \cite{goyal2001quantized,kovacevic2007life,rath2004frame,casazza2003equal},
we explore {\it systematic tight DFT frames}.
Although it is straightforward to construct systematic DFT frames, we prove that systematic
``tight'' DFT frames exist only for specific DFT frames.
More precisely, we show that a systematic frame is tight if and only if
data (systematic) samples are circularly equally spaced, in the codewords
generated by that frame.
When such a frame does not exist, we will be looking for systematic DFT frames with the ``best''
performance, from the minimum mean-squared reconstruction error sense. 
We also demonstrate which systematic frames are the ``worst'' in this sense.
In addition, we show that circular shift and reversal of the vectors
in a DFT frame does not change the eigenvalues of
the frame operator. We use these properties to categorize different
systematic frames of an $(n,k)$ DFT frame
based on their performance.

Another main contribution of this paper is to find bounds on the {\it extreme eigenvalues} of
$V^HV$, where $V$ is a square or non-square subframe of a DFT frame.
The properties of the eigenvalues of such frames are central to establish
many of the result in this paper.
These bounds are used to determine the conditions
required for a systematic frame so as to be tight.
Besides, eigenvalues are crucial in establishing the best and worst systematic frames.

The paper is organized as follows. In Section~\ref{sec:def}, we present the basic definitions
and a few fundamental lemmas that will be used later in the paper.
In Section~\ref{sec:DFTcodes}, we introduce DFT frames and set
the ground to study the extreme eigenvalues of their subframes.
Section~\ref{sec:sys} motivates the work in this paper by introducing
systematic DFT frames and their application. Some result on the the extreme
eigenvalues of DFT frames and their subframes are presented in Section~\ref{sec:eig}.
Sections~\ref{sec:perform} and \ref{sec:class} is devoted to the evaluation of reconstruction error
and classification of systematic frames based on that. We conclude in Section~\ref{sec:con}.

For notation, we use boldface lower-case letters for vectors, boldface upper-case letters for
matrices,  $(.)^T$ for transpose, $(.)^H$ for conjugate transpose, $(.)^\dagger$ for 
pseudoinverse, $(.)^\ast$ for conjugate, $\tr(.)$ for the trace, $\mathbb{E(.)}$ for the mathematical expectation, and $\| . \|$ for
the Euclidean norm.
The dimensions of square and rectangular matrices are indicated, respectively, by one and two subscripts when required.

\section{Definitions and preliminaries}
\label{sec:def}

In this section, we introduce the definitions and some
basic results which are frequently used in the paper.
\begin{defn}
A spanning family of $n$ vectors $F = \{\bm{f}_i\}_{i=1}^{n}$ in a $k$-dimensional complex vector space
${\mathbb C^k}$ is called a {\it frame} if there exist $0 <a \leq b $
such that for any $\bm{x} \in {\mathbb C^k}$
\begin{align}
a \| \bm{x} \|^2 \leq \sum_{\substack i=1}^{n}{| \langle \bm{x} , \bm{f}_i  \rangle |^2} \leq b \| \bm{x} \|^2,
 \label{eq:frame}
 \end{align}
where $\langle \bm{x} , \bm{f}_i  \rangle$ denotes the inner
product of $\bm{x}$ and $\bm{f}_i $ and gives the $i$th coefficient for the frame
expansion of $\bm{x}$ \cite{kovacevic2007life,casazza2003equal,kovacevic2008introduction}.
$a$ and $b$ are called {\it frame bounds}; they, respectively, ensure that the vectors span the space,
and the basis expansion converges.
A frame is {\it tight} if $a=b$.
{\it Uniform} or {\it equal-norm} frames are frames with same norm for all elements, i.e.,
$\|\bm{f}_i\|=\|\bm{f}_j\|$, for $i,j = 1,\hdots,n$. 

\end{defn}

\begin{defn}
An $n\times n$ Vandermonde matrix with unit complex entries is defined by
\begin{align}
W \triangleq \frac{1}{\sqrt{n}}\left( \begin{array}{ccccccc}
       1   &1 & \cdots & 1 \\
       e^{j\theta_1}   &e^{j\theta_2} & \cdots & e^{j\theta_n} \\
       \vdots & \vdots & \ddots & \vdots \\
       e^{j(n-1)\theta_1}   &e^{j(n-1)\theta_2}  & \cdots &e^{j(n-1)\theta_n}  \\
      \end{array}
\right),  \label{eq:Vand2}
\end{align}
in which $\theta_p \in [0, 2 \pi) $ and $\theta_p \neq \theta_q $ for $p\neq q$, $1\leq p,q \leq n$.
If $\theta_p = \frac{2\pi}{n}(p-1)$, $W$ becomes the well-known IDFT matrix \cite{mitra1998digital}.
For this Vandermonde matrix we can write \cite{tucci2011eigenvalue2}, \cite{tucci2011eigenvalue}
\begin{align}
 \det(WW^H) = |\det(W)|^2 =\frac{1}{n^n}\prod_{\substack 1\leq p<q\leq n}{|e^{i\theta_p}-e^{i\theta_q}|^2}.
 \label{eq:detV0}
 \end{align}
\end{defn}

Central to this work is the properties of the {\it eigenvalues} of $V^HV$ or $VV^H$,
in which $V$ is a submatrix of a DFT matrix.\footnote{Note that eigenvalues of $V^H V$ and $VV^H$
are equal for a square $V$; also, $V^H V$ and $VV^H$ have the same nonzero eigenvalues for a non-square $V$.}
Hence, we recall some bounds on the eigenvalues of Hermitian matrices which are used in this paper.
Let $A$ be a Hermitian $k\times k$ matrix with real eigenvalues
$\{\lambda_{1}(A), \hdots, \lambda_{k}(A)\}$ which are collectively called the {\it spectrum} of $A$, and
assume $\lambda_{1}(A) \geq \lambda_{2}(A)\geq \cdots \geq \lambda_{k}(A)$.
Schur-Horn inequalities show to what extent the eigenvalues of
a Hermitian matrix constraint its diagonal entries.
\begin{pro}\textit{Schur-Horn inequalities} \cite{seber2008matrix} \\
Let $A$ be a Hermitian $k\times k$ matrix with real eigenvalues
$\lambda_{1}(A) \geq \lambda_{2}(A)\geq \cdots \geq \lambda_{k}(A).$
Then,  for any $1 \leq i_1 < i_2< \cdots < i_l \leq k$,
\begin{align}
 \lambda_{k-l+1}(A) + \cdots + \lambda_{k}(A) &\leq
 a_{i_1i_1}+ \cdots + a_{i_li_l} \nonumber \\
 &\leq  \lambda_{1}(A) + \cdots + \lambda_{l}(A),
 \label{eq:Schur-Horn}
 \end{align}
 where $ a_{11}, \hdots, a_{kk}$ are the diagonal entries of $A$.
Particularly, for $l=1$  and $l=k$ we obtain
 \begin{align}
 \label{eq:Schur-Horn1}
 \lambda_{k}(A) \leq  a_{11} \leq  \lambda_{1}(A), \\
  \sum_{\substack i=1}^{k}{\lambda_{i}(A)} = \sum_{\substack i=1}^{k}{a_{ii}}.
 \label{eq:Schur-Horn2}
 \end{align}
\end{pro}

 Another basic question in linear algebra asks the degree to which the eigenvalues
 of two Hermitian matrices constrain the eigenvalues of their sum.
Weyl's theorem gives an answer to this question in the following set of inequalities.

\begin{pro}{Weyl inequalities \cite{seber2008matrix}} \\
Let $A$ and $B$ be two Hermitian $k\times k$ matrices with spectrums $\{\lambda_1(A),\ldots,\lambda_k(A)\}$ and $\{\lambda_1(B),\ldots,\lambda_k(B)\}$, respectively.
Then, for $i,j\leq k$, we have
\begin{align}
\label{eq:Weyl1}
 \lambda_{i}(A+B) \leq  \lambda_{j}(A) + \lambda_{i-j+1}(B) \qquad \text{for} \quad j \leq i, \\
 \lambda_{i}(A+B) \geq  \lambda_{j}(A) + \lambda_{k+i-j}(B) \qquad \text{for} \quad j \geq i.
 \label{eq:Weyl2}
 \end{align}
 \end{pro}

\begin{cor}
 If  $A+B=\gamma I_k, \gamma >0,$ where $A$ and $B$ are Hermitian matrices, then $\lambda_{j}(A) + \lambda_{k-j+1}(B) = \gamma.$
 \label{cor1}
\end{cor}
 \begin{proof}
 It suffice to set $i=k$ and $i=1$ respectively in \eqref{eq:Weyl1} and \eqref{eq:Weyl2},
 and use $ \lambda_{k}(A+B) = \lambda_{1}(A+B)=\gamma $ which is obtained from $A+B=\gamma I_k$.
 \end{proof}

\begin{lem}
Let $A$ and $B$ be two Hermitian $k\times k$ matrices
and suppose that, for every $1 \leq i,j \leq k$, $A_{i,j} =e^{j\theta_i}B_{i,j}$; then
$A^HA$ and $B^HB$ have the same spectrum.
\label{lem0}
\end{lem}

\begin{proof}
The proof is immediate using Lemma~3 \cite{tucci2011eigenvalue} since
$(A^HA)_{i,j}=\frac{e^{j\theta_i}}{e^{j\theta_i}}(B^HB)_{i,j}$; i.e., $A^HA=B^HB.$
  \end{proof}

\section{DFT Frames}
\label{sec:DFTcodes}
\subsection{BCH-DFT Codes}

BCH-DFT codes \cite{Marshall} are {\it linear block codes} over the complex field
whose parity-check matrix $H$ is defined based on the DFT matrix; they insert $d-1$ cyclically adjacent
zeros in the frequency-domain function (Fourier transform) of any codeword where $d$ is the designed
distance of that code \cite{blahut2003algebraic}.
Real BCH-DFT codes, a subset of complex BCH-DFT codes, benefit from
a generator matrix with real entries.
The generator matrix of an $(n,k)$ {\it real BCH-DFT
code}\footnote{Real BCH-DFT codes do not exist  when $n$ and $k$ are both even \cite{Marshall}.}
is typically defined by \cite{gabay2007joint,rath2004subspace,rath2004frame,vaezi2011DSC}
\begin{align}
G= \sqrt{\frac{n}{k}} W_n^H \Sigma W_k,
\label{eq:G1}
\end{align}

\noindent in which $W_l$ represents the DFT matrix of size $l$, and
$\Sigma$ is defined as
\begin{align}
\Sigma_{n\times k} = \left( \begin{array}{ccccccc}
       I_\alpha  & \bm{0}  \\
       \bm{0}   & \bm{0}  \\
      \bm{0}    &   I_\beta    \\
      \end{array}
\right), \label{eq:cov}
\end{align}
where $\alpha = \lceil n/2\rceil  -\lfloor (n-k)/2\rfloor$, $\beta=k-\alpha$,
and the sizes of zero blocks are such that $\Sigma$ is an $n\times k$ matrix \cite{gabay2007joint}.
One can check that $\Sigma^H \Sigma= I_k$, and $\Sigma \Sigma^H $ is an $n\times n $ matrix given by
\begin{align}
\Sigma\Sigma^H = \left( \begin{array}{ccccccc}
       I_\alpha  & \bm{0} & \bm{0}  \\
       \bm{0}   & \bm{0} & \bm{0} \\
      \bm{0}    & \bm{0} &   I_\beta    \\
      \end{array}
\right). \label{eq:cov1}
\end{align}
Note that, having $n-k$ consecutive zero rows, $\Sigma$ inserts $n-k$ consecutive zeros to
each codeword in the frequency domain which ensures having a BCH code \cite{blahut2003algebraic,Marshall}.

\begin{rem}
Removing $W_k$ from \eqref{eq:G1} we end up with a complex $G$, representing a {\it complex BCH-DFT code}. In such a code,
$\alpha$ and $\beta$ can be any nonnegative integers such that $\alpha + \beta =k$.
\label{rem1}
\end{rem}
The parity-check matrix $H$, both in real and complex codes, consist of the $n-k$ columns of $W_n^H$
corresponding to the zero rows of $\Sigma$; thus, $HG=0.$

\subsection{Connection to Frame Theory}
The generator matrix $G$ in \eqref{eq:G1} can be viewed as an {\it analysis frame operator}.
In this view, a real BCH-DFT code is a rotation of the well-known {\it harmonic frames} \cite{casazza2003equal,kovacevic2008introduction},
and a complex BCH-DFT code is basically a harmonic frame.
The latter can be understood by removing $W_k$ from \eqref{eq:G1} which results in
a complex BCH-DFT code, on the one hand, and the analysis frame operator of
a harmonic frame, on the other hand. The former is then evident as $W_k$ is a rotation matrix.
Further, it is easy to see that
the {\it frame operator} $G^HG$ and {\it Gramian} $GG^H$ are equal to
\begin{align}
\label{eq:GHG}
G^HG &= \frac{n}{k} I_k,\\
GG^H&= \frac{n}{k} W_n^H \Sigma \Sigma^HW_n.
\label{eq:GGH}
\end{align}

The following lemma presents some properties of the frame operator and relevant matrices which are
crucial for our results in this paper. 

\begin{lem}
Let $ G_{p\times k}$ be a matrix consisting of $p$ arbitrary rows of $G$ defined by \eqref{eq:G1}.
Then, the following statements hold:\\
i. $GG^H$ is a Toeplitz and circulant matrix\\
ii. $ G_{p\times k}G_{p\times k}^H, 1< p < n$ is a Toeplitz matrix\\
iii. All principal diagonal entries of $ G_{p\times k}G_{p\times k}^H, 1\leq p \leq n$ are equal to $1$.
\label{lem1}
\end{lem}

\begin{proof}
Let $a_{r,s}$ be the $(r,s)$ entry of the matrix $GG^H$ then it can readily be shown that
\begin{equation}
\begin{aligned}
a_{r,s}
&=\frac{1}{k}\sum_{m=0}^{\alpha-1}e^{jm(\theta_r-\theta_s)} +\frac{1}{k} \sum_{m=n-\beta}^{n-1}e^{jm(\theta_r-\theta_s)},\\
\end{aligned}
\label{eq:Toep}
\end{equation}
in which $\theta_x = \frac{2\pi}{n}(x-1)$. From this equation, it is clear that $a_{r,s}=a_{r+i,s+i}$;
that is, the elements of each diagonal are equal, which means that  $GG^H$ is a {\it Toeplitz} matrix.
In addition, we can check that $a_{r,n}=a_{r+1,1}$, i.e., the last entry in each row is equal to the first entry of the next row.
This proves that the Toeplitz matrix $GG^H$ is {\it circulant} as well \cite{gray2006toeplitz}.
Also, a quick look at \eqref{eq:Toep} reveals that the elements of the principal diagonal $(r=s)$ are equal to $1$.  Similarly, one
can see that for any  $1 < p < n$, the square matrix $G_{p\times k}G_{p\times k}^H$ is also a Toeplitz matrix; it is
not necessarily circulant, however.
\end{proof}

Considering Remark~\ref{rem1}, one can check that \eqref{eq:Toep} is also valid for complex BCH-DFT codes.
Note that, $\alpha$ and $\beta$ are less constrained for these codes, as mentioned in
Remark~\ref{rem1}.

\begin{rem}
Lemma~\ref{lem1} also holds for complex BCH-DFT codes.
\end{rem}
Further, in a DFT frame, in general, the $n-k$ zero rows of $\Sigma$ are not required
to be successive if they are not designed for error correction.
That is any matrix that can be rearranged as $[I_k \;|\;\bm{0}_{k \times n-k}]^T$ may represent $\Sigma$.
Then, $\Sigma\Sigma^H$ is not necessarily in the form given in \eqref{eq:cov1};
it can be any square matrix of size $n$ with $k$ nonzero elements equal to 1, arbitrary located on the main diagonal.
Then, again Lemma~\ref{lem1}
holds because $a_{rs}=\frac{1}{k}\sum_{i=0}^{k-1}e^{jm_i(\theta_r-\theta_s)}$ and $m_i \in \{1,\hdots,n\}$.
\begin{rem}
Lemma~\ref{lem1} holds for all DFT frames.
\label{rem3}
\end{rem}

\section{Systematic DFT Frames}
\label{sec:sys}

In general, every sample in the codewords of a DFT frame is a linear
combination of all data samples of the input block, i.e., the data samples do not appear explicitly in the codewords.
A specific method of encoding, known as {\it systematic encoding}, leaves the data samples unchanged.
These unchanged samples can be exhibited in any component of the codeword, therefore:

\begin{defn}
An $(n,k)$ frame is said to be systematic if its analysis frame operator
includes $I_{k}$ as a subframe.
\label{def3}
\end{defn}

\subsection{Motivation and Applications}
\label{sec:mot}
In the context of channel coding, there is a special interest in
{\it systematic codes}  \cite{blahut2003algebraic} 
 since the input data is embedded in the encoded output
which simplifies the encoding and decoding algorithms. For example, in systematic
convolutional codes data can be read directly if no errors are made,
or in case only parity bits are affect in an {\it erasure channel}.
Systematic codes are also used in parity-based distributed source coding (DSC)
techniques, e.g., DSC that uses turbo codes for compression
\cite{bajcsy2001coding,garcia2001compression,aaron2002compression}.
DSC addresses the problem of compressing correlated sources by separate encoding and
joint decoding and has found application in sensor networks and video compression \cite{dragotti2009distributed}.
The compression is usually realized through the use of {\it binary channel codes}.

Recently, the authors have introduced a new framework that exploits {\it real-number codes}
for DSC \cite{vaezi2011DSC}
and distributed joint source-channel coding (DJSCC) \cite{vaezi2012WZ}.
Specifically, by using BCH-DFT codes it has been shown that this framework
can result in a better compression compared to the conventional one.
There are  syndrome- and parity-based approaches to do DSC \cite{vaezi2011DSC,dragotti2009distributed,vaezi2012WZ};
the compression is achieved by representing the input data with fewer samples,
which are a linear combination of the input samples.
To do so, in the former approach the encoder
generates syndrome samples with respect to a DFT code, whereas it
 generates parity samples with respect to a {\it systematic} DFT code in the
latter case.
The parity (syndrome) is then quantized and transmitted over a 	
noiseless channel. Assuming the asymmetric DSC \cite{WZ}, where one source is available at the decoder
as side information, the decoder looks for the closest vector to the side information,
among the vectors whose parity (syndrome) is equal to the received one.

The parity-based approach is worthwhile as the parity of a real DFT code
is a real vector contrary to its syndrome which is complex.
More importantly, to accomplish DJSCC only the parity-based approach is
known to be applicable \cite{vaezi2012WZ}. On the other hand, the parity-based approach mandates
systematic DFT codes and is the main motivation of this work.

\subsection{Construction}
\label{sec:con}

In view of Definition~\ref{def3}, the systematic
generator matrix for a real BCH-DFT code can be obtained by
  \begin{align}
  G_{\mathrm{sys}}=GG_{k}^{-1},
\label{eq:Gsys2}
\end{align}
in which $G_{k}$ is a submatrix (subframe \cite{rath2004frame}) of $G$ including $k$ arbitrary rows of $G$.
Note that $G_{k}$ is invertible since it can be represented as
\begin{align}
 G_{k} =\sqrt{\frac{n}{k}} W_{k \times n}^H \Sigma W_{k} = V_{k}^H W_{k},
\label{eq:Gsys3}
\end{align}
in which $V_{k}^H \triangleq \sqrt{\frac{n}{k}}W_{k \times n}^H \Sigma$ and $W_{k}$
are invertible as they are Vandermonde and DFT matrices, respectively.
Obviously, this argument is valid if $W_{k}$ is removed and/or when the $n-k$
zero rows of $\Sigma$ are not successive.
This indicates that any $k$ rows of a  DFT frame make a {\it  basis}
of ${\mathbb C^k}$ and proves that $G_{k}^{-1}$ and thus
systematic DFT frames exist for any DFT frame.

\begin{rem}
 From the above discussion and Remark~\ref{rem3} one can see that what
  we prove in the remainder of this paper is valid for ``any''
DFT frame, not just for real BCH-DFT codes.
\end{rem}

The construction in \eqref{eq:Gsys2} suggests that for each DFT frame there are many (but, a finite number of)
systematic frames  since the rows of $G_{k}$ can be arbitrarily chosen
from those of $G$. This will be discussed in detail later in Section~\ref{sec:number}.
The codewords generated by these systematic frames differ in the
``position'' of systematic samples (i.e., input data).
This implies that  parity (data) samples are
not restricted to form a consecutive block in the associated codewords.
Such a degree of freedom is useful in the sense that one can find the
most suitable systematic frames for specific applications (e.g., the
one with the smallest reconstruction error.)

\subsection{Optimality Condition}
\label{sec:mot}

From rate-distortion theory, it is well known that the rate required to transmit a source, with a given distortion, increases as
the variance of the source becomes larger \cite{Cover}. Particularly, for Gaussian sources this relation
is logarithmic with variance, under the mean-squared error (MSE) distortion measure.
In DSC that uses real-number codes \cite{vaezi2011DSC}, since coding is performed before quantization,
the variance of transmitted sequence depends on the behavior of the
encoding matrix. In syndrome approach, $\bm{s}=H\bm{x}$ \cite{vaezi2011DSC} and
it can be checked that $\sigma_{\bm{s}}=\sigma_{\bm{x}}$, that is,
the variance is preserved.\footnote{In general, any unitary matrix $U$ preserves norms, i.e., for
any complex vector $\bm{x}$, $\|U\bm{x}\|=\|\bm{x}\|.$ Note that $H$ is not unitary because
it is not a square matrix; however, its rows are selected from a unitary matrix and are orthonormal.
This lead to $HH^H=I_{n-k}$, and $\tr(H^HH)=n-k$.} However, as we show shortly,
this is not valid in parity approach and the variance of parity samples
depends on the behavior of encoding matrix $G_{\mathrm{sys}}$.
 In view of rate-distortion theory, it makes a lot of sense
to keep this variance as small as possible. Not surprisingly, we will show that using a tight
frame (tight $G_{\mathrm{sys}}$) for encoding is optimal.

Let $\bm{x}$ be the message vector, a column vector whose elements are i.i.d. random variables with variance $\sigma^2_x$,
 and let $\bm{y}=G_{\mathrm{sys}}\bm{x}$ represent the codeword generated using the systematic frame.
The variance of  $\bm{y}$ is then given by

  \begin{equation}
\begin{aligned}
\sigma^2_y &=  \frac{1}{n} \mathbb{E}\{\bm{y}^H\bm{y}\} =  \frac{1}{n} \mathbb{E}\{ \bm{x}^H G_{\mathrm{sys}}^H  G_{\mathrm{sys}}\bm{x}  \}  \\
&= \frac{1}{n} \sigma^2_x \tr{ (G_{\mathrm{sys}}^H  G_{\mathrm{sys}})},
\label{eq:vary}
\end{aligned}
\end{equation}
and
 \begin{equation}
\begin{aligned}
 \tr\left(G_{\mathrm{sys}}^HG_{\mathrm{sys}}\right) &= \tr\left(G_k^{-1H}G^HGG_k^{-1}\right)  \\
  &= \frac{n}{k}\tr\left((G_k G_k^H)^{-1}\right)  \\
  &= \frac{n}{k} \tr\left((V_{k}^HV_{k})^{-1}\right) \\
  &= \frac{n}{k}\sum_{i=1}^k\frac{1}{\lambda_i},
\label{eq:G7}
\end{aligned}
\end{equation}
in which $\lambda_1 \geq \lambda_2 \geq \cdots \geq \lambda_k>0$ are the eigenvalues of $G_{k}G_{k}^H$ (or $V_{k}^HV_{k}$ equivalently).

This shows that the variance of codewords, generated by a systematic frame, depends on the submatrix $G_{k}$ which is used to create $G_{\mathrm{sys}}$. $G_{k}$, in turn,
is fully known once the position of systematic samples is fixed in the codewords. In other words, the ``position'' of systematic samples determines
the variance of the codewords generated by a systematic DFT frame.
To minimize the effective range of transmitted signal, from \eqref{eq:vary} and \eqref{eq:G7}, we need to do the following optimization problem
\begin{equation}
\begin{aligned}
& \underset{\lambda_i}{\text{minimize}}
& & \sum_{i=1}^k\frac{1}{\lambda_i}\\
& \text{s.t.}
& & \sum_{i=1}^k\lambda_i=k, \,\, \lambda_i>0,
\end{aligned}
\label{eq:Omin}
\end{equation}
where, the constraint $\sum_{i=1}^k\lambda_i=k$ is achieved in consideration of Lemma~\ref{lem1} and \eqref{eq:Schur-Horn2}.

By using the Lagrangian method \cite{boyd2004convex}, we can show that the optimal eigenvalues are $\lambda_i=1$; this implies a tight frame \cite{goyal2001quantized}.
In the sequel, we analyze
the eigenvalues of $G_{p\times k}G_{p\times k}^H$, $p \leq n$, that helps us characterize tight systematic frames, so as to
minimize the variance of transmitted codewords.

\section{Main Results on the Extreme Eigenvalues}
\label{sec:eig}

In this section we
investigate some bounds on the eigenvalues of $G_{p\times k}G_{p\times k}^H$ where $G$ is defined in \eqref{eq:G1}.
These bounds play an important role in the performance evaluation of the systematic DFT frames.
We also determine the exact values of some eigenvalues in certain cases.

\begin{thm}
Let $G_{p\times k}, 1 \leq p \leq n$ be any $p\times k$ submatrix of $G$. Then, the smallest eigenvalue of
$G_{p\times k}G_{p\times k}^H$ is no more than one, and the largest eigenvalue of
$G_{p\times k}G_{p\times k}^H$ is at least one.
\label{thm1}
\end{thm}

\begin{proof}
From Lemma~\ref{lem1}, we know that all principal diagonal entries of $G_{p\times k}G_{p\times k}^H$
 are unity. As a result, using the Schur-Horn inequality in \eqref{eq:Schur-Horn1},
 we obtain $\lambda_{\min}(G_{p\times k}G_{p\times k}^H) \leq 1 \leq  \lambda_{\max}(G_{p\times k}G_{p\times k}^H)$.
 This proves the claim.
\end{proof}

Note that  $\lambda_{1}(G_{p\times k}G_{p\times k}^H)=\lambda_{1}(G_{p\times k}^HG_{p\times k})$
for any $G_{p\times k}$. Nevertheless,
 this is not correct for $\lambda_{\min}$ in general.
 A tighter bound on $\lambda_{1}$
 can be achieved when $G_{p\times k}$ is a tall\footnote{An $m\times n$ matrix $A$ is called to be tall if $m>n$. Similarly,
 if $m<n$, then $A$ is a short matrix. } matrix.

\begin{thm}
Given a tall (short) $G_{p\times k}$, the largest (smallest) eigenvalue of
$G_{p\times k}^HG_{p\times k}$ is lower (upper) bounded by $p/k$.
\label{thm0}
\end{thm}

\begin{proof} Let $p> k$. Since all diagonal entries of $G_{p\times k}G_{p\times k}^H$
 are unity, from \eqref{eq:Schur-Horn2} we have $\sum_{\substack i=1}^{p}{\lambda_{i}(G_{p\times k}G_{p\times k}^H)} =p$. On the other hand,
 since the nonzero eigenvalues of $G_{p\times k}G_{p\times k}^H$ and $G_{p\times k}^HG_{p\times k}$ are
 equal, $G_{p\times k}G_{p\times k}^H$ has $k$ nonzero eigenvalues and we get
 \begin{equation}
\begin{aligned}
 p &=\sum_{\substack i=1}^{p}{\lambda_{i}(G_{p\times k}G_{p\times k}^H)}\\
 &=\sum_{\substack i=1}^{k}{\lambda_{i}(G_{p\times k}^HG_{p\times k})}  \\
 &\leq k\lambda_{1}(G_{p\times k}^HG_{p\times k}).
 \label{eq:tall}
 \end{aligned}
 \end{equation}
Thus, for any tall $G_{p\times k}$, $\lambda_{1}(G_{p\times k}^HG_{p\times k})=\lambda_{1}(G_{p\times k}G_{p\times k}^H)\geq \frac{p}{k}>1$.
Following a similar line of proof, for a short submatrix $(p<k)$ we obtain $\lambda_{\min}(G_{p\times k}^HG_{p\times k})\leq \frac{p}{k}<1$.

Obviously the same bounds are valid for the extreme eigenvalues of $G_{p\times k}G_{p\times k}^H$.
What is more, since $p/k$ is the average value of eigenvalues, considering that $\lambda_{\min}(G_{p\times k}G_{p\times k}^H)= 0$
for $p > k$, and $\lambda_{\min}(G_{p\times k}^HG_{p\times k})= 0$
for $p < k$, from \eqref{eq:tall} we conclude that corresponding bounds on the largest eigenvalues are strict.
\end{proof}
It is worth noting that in \eqref{eq:tall} the equality is achieved when $p=n$; it can also be achieved for ``specific''
submatrices only in the case of integer oversampling, i.e., when $n=Mk$, as we discuss later in this paper.

 We use the above results to find better bounds for the extreme eigenvalues of $G_{k}G_{k}^H$ in the following theorem.

\begin{thm}
For any $G_{k}$, a square submatrix of $G$ in \eqref{eq:G1} in which $n\neq Mk$, the smallest (largest) eigenvalue of
$G_{k}G_{k}^H$ is strictly upper (lower) bounded by $1$.
\label{thm2}
\end{thm}
\begin{proof}
See Appendix~\ref{sec:App1}.
\end{proof}

Theorem~\ref{thm2} implies that for $n\neq Mk$ we cannot have ``tight'' systematic frames.
Because, for a frame with frame operator $F^HF$, the tightest possible frame bounds  are, respectively,
$a=\lambda_{\min}(F^HF)$ and $b=\lambda_{\max}(F^HF)$ \cite{chen2011mathematical}.
In other words, for a tight frame
$\lambda_{\min}(F^HF)=\lambda_{\max}(F^HF)$; i.e., the eigenvalues of $F^HF$ are equal \cite{goyal2001quantized}.

\begin{cor}
\label{cor3}
Tight systematic DFT frames can exist only if $n = Mk$, where $M$ is a positive integer.
\end{cor}

\noindent Note that systematic DFT frames are not necessarily tight for $n = Mk$. In Section~\ref{sec:class},
we prove that tight systematic DFT frames exist for $n = Mk$ and show how to construct such frames.

In the remainder of this section, we shall find exact values, rather than bounds, for
some of the eigenvalues of $G_{k}^H G_{k}$ when $k < n \leq 2k$. This range of $n$ is specifically
important in parity-based DSC \cite{vaezi2011DSC}, where $n-k$
parity samples are used to represent $k$ samples and so for compression, $n-k < k$.

\begin{thm}
For any $G_{k}$, a square submatrix of $G$ in \eqref{eq:G1}, where $k<n<2k$, the $2k-n$ largest eigenvalues of $G_{k}G_{k}^H$ are equal to $n/k$.\\
\label{thm4}
\end{thm}
\begin{proof}
From Corollary~\ref{cor1} we know that if two Hermitian matrices sum up to a scaled identity matrix, their eigenvalues add up to be fixed.
Thus, if $A$ and $B$ have the same spectrum we obtain
 \begin{align}
\lambda_{j}(A) + \lambda_{k-j+1}(A) = \gamma.
 \label{eq:f2}
 \end{align}
Now, let $ G$ be partitioned as
$ G = \left[\begin{array}{c}
      G_{ k} \\ 
      \bar{G}_{p \times k}
 \end{array}\right]$ where $p=n-k$.
Let $A=G_{ k}^H G_{ k}$ and $ B=\bar{G}_{p \times k}^H \bar{G}_{p \times k}$,
then $A+B=G^H G = \frac{n}{k}I_k$.
Clearly, Corollary~\ref{cor1} holds with $\gamma = \frac{n}{k}$.
 Also, note that when $p < k$ then $\bar{G}_{p\times k}^H\bar{G}_{p\times k}$
  has only $p$ nonzero eigenvalues. Therefore, in such a case, $k-p$ largest eigenvalues
  of $G_{k}^HG_{k}$ are equal to $n/k$. 
  \end{proof}

 Another interesting case arises when $n=2k$.
  Numerical results shows that under this condition, $A$ and $B$
  have the same set of eigenvalues. We prove this when
  $G_{k}$ either includes successive or every other rows of $G$.
 In such cases, one can verify that $(\bar{G}_{k})_{i,j} =  e^{j\theta} (G_{k})_{i,j}$;
 thus, Lemma~\ref{lem0} holds and $A$ and $ B$
  have the same eigenvalues.
 Hence, from  \eqref{eq:f2} we get
 \begin{align}
\lambda_{j}(G_{ k}^H G_{ k}) + \lambda_{k-j+1}(G_{ k}^H G_{ k}) = \frac{n}{k} =2.
 \label{eq:f3}
 \end{align}
This further implies that for odd values of $k$ the middle eigenvalue of $G_{ k}^H G_{ k}$ is 1.

We close this section with an example illustrating some of the above properties.
 Consider an $(n, k)$ DFT frame and the the following two cases. First, the rows of $G_{ k}$ are evenly spaced rows of $G$ (i.e., either odd rows or even rows).
  This is the ``best'' submatrix in the sense that it minimizes the MSE. For such a submatrix, all eigenvalues are known to be equal,
  as it is a DFT matrix. For example, for $n=10, k=5$, the best square submatrix results
  in $\lambda=1$ with multiplicity of 5. The other extreme case, which maximizes the MSE,
  happens when the rows of $G_{ k}$ are circularly consecutive rows of $G$.
  Again, for the above example, $\lambda= \{ 0.0011, 0.1056, 1,  1.8944, 1.9989\}$.
 With these examples in mind, we will explore the best and worst frames in Section~\ref{sec:class}.
 We shall now discuss signal reconstruction   for systematic frames.

\section{Performance Analysis}
\label{sec:perform}

In this section, we analyze the performance
of quantized systematic DFT codes using the quantization model proposed in \cite{goyal2001quantized},
which assumes that noise components are uncorrelated and each noise component $q_i$
has mean $0$ and variance $\sigma_q^2$, i.e., for any $i, j$,
 \begin{align}
 \mathbb{E}\{q_i\} =0, \qquad \mathbb{E}\{q_iq_j \} =\sigma_q^2 \delta_{ij}.
 \label{eq:Q}
\end{align}
For one thing, $q$ can be uniformly distributed on $[-\Delta/2, \,   \Delta/2 ],$ where $\sigma_q^2 =\Delta^2/12$.
We assume the quantizer range covers the
dynamic range of all codewords encoded using the systematic DFT code in \eqref{eq:Gsys2}.

Let $\bm{x}$ be the signal (message) to be transmitted.
The corresponding codeword is generated by
\begin{align}
\bm{y}=G_{\mathrm{sys}}\bm{x}.
\label{eq:y}
\end{align}
This is then quantized to $\hat{\bm{y}}$ and transmitted. Assuming the quantization model in
\eqref{eq:Q}, transmitted codeword can be modeled by
\begin{align}
\hat{\bm{y}}  = G_{\mathrm{sys}}\bm{x}+ \bm{q},
\label{eq:yhat}
\end{align}
 where $\bm{q}$ represents quantization error. This also models the received codvector
 provided that there is no error or erasure in channel.
 Now, suppose we want to estimate \mbox{\boldmath$x$}
from \eqref{eq:yhat}. This can be done through the use
of linear or nonlinear operations.

\subsection{Linear Reconstruction}
\label{subsec:linear}

We first consider {\it linear reconstruction} of \mbox{\boldmath$x$} form
$\hat{\bm{y}}$ using the pseudoinverse \cite{goyal2001quantized} of $G_{\mathrm{sys}}$,
which is defined by
\begin{align}
G_{\mathrm{sys}}^\dagger = (G_{\mathrm{sys}}^HG_{\mathrm{sys}})^{-1}G_{\mathrm{sys}}^H = \frac{k}{n}G_{k}G^H.
\label{eq:Gdagger}
\end{align}
The linear reconstruction is hence given by
\begin{equation}
\begin{aligned}
\hat{\bm{x}} = \frac{k}{n}G_{k}G^H \hat{\bm{y}}=\bm{x}+\frac{k}{n}G_{k}G^H\bm{q},
\label{eq:G3}
\end{aligned}
\end{equation}
where $\bm{q}$ represents quantization error.

Let us now evaluate the reconstruction error.
The mean-squared reconstruction error, due to the quantization noise,
using a systematic frame can be written as
\begin{equation}
\begin{aligned}
\mathop{\mathrm{MSE_{q}}} &= \frac{1}{k} \mathbb{E}\{\|\hat{\bm{ x}} -\bm{ x}\|^2\}
= \frac{1}{k}  \mathbb{E}\{\|G_{\mathrm{sys}}^{\dagger}\bm{q}\|^2\}  \\
&= \frac{1}{k} \mathbb{E}\{\bm{q}^H G_{\mathrm{sys}}^{\dagger H}G_{\mathrm{sys}}^{\dagger}\bm{q}\}
= \frac{1}{k}  \sigma_q^2 \tr\left(G_{\mathrm{sys}}^{\dagger H}G_{\mathrm{sys}}^{\dagger}\right)\\
&= \frac{k}{n^2} \sigma_q^2 \tr\left(GG_{k}^H G_{k}G^H\right)\\
&= \frac{k}{n^2} \sigma_q^2 \tr\left(G_{k}^H G_{k}G^HG\right)\\
&= \frac{1}{n}  \sigma_q^2 \tr\left(G_{k}^H G_{k}\right)= \frac{k}{n} \sigma_q^2 ,
\label{eq:G6}
\end{aligned}
\end{equation}
where the last step follows because of Lemma~\ref{lem1}. This shows that
DFT codes reduce quantization error.

The fact that the MSE is inversely proportional
to the {\it redundancy} of the frame is a well-known result for
{\it tight} frames 
\cite{kovacevic2007life,Vaezi2011LS,goyal2001quantized, rath2004frame}.
The above analysis, however, indicates that the
MSE is the same for all systematic DFT frames of the same size, no matter they are tight or not.
This is yet assuming that the effective range of the codewords generated by different $G_{\mathrm{sys}}$ is equal,
which implies the same $\sigma_q^2$ for a given number of quantization levels.
However, from  \eqref{eq:vary} it is known that, for a fixed number of quantization levels, $\sigma_q^2$ depends on the
variance of transmitted codewords ($\sigma_y^2$) if the quantizer is designed to
cover the entire effective range of codewords.
Obviously, though, $\sigma_y^2$ can vary from one systematic frame to another, as shown in \eqref{eq:G7}.

\begin{thm}
When encoding with a systematic DFT frame in \eqref{eq:Gsys2} and decoding with linear reconstruction,
for the noise model \eqref{eq:Q} and given a same number of quantization levels,
 the MSE is minimum if and only if the systematic frame is tight.
\label{thm5}
\end{thm}

\begin{proof}
All systematic DFT frames amount to a same quantization error provided that
the effective range of codewords are fully covered, as shown in \eqref{eq:G6}.
Nevertheless, for a fixed number of quantization levels more codewords
are within the range of quantizer if the systematic frame is tight.
This is clear from \eqref{eq:G7} and \eqref{eq:Omin},
recalling that \eqref{eq:Omin} is minimized by the tight frames.
Moreover, any frame that minimizes \eqref{eq:Omin} is required to be tight.
This will be proved in Section~\ref{sec:class1}.
\end{proof}

The problem we are considering in Theorem~\ref{thm5} is somewhat the {\it dual}
of Theorem 3.1 in \cite{kovacevic2007life}. Note that in \cite[Theorem 3.1]{kovacevic2007life}
``uniform" frames are used for encoding which implies the same variance for all samples of
 codewords whereas the reconstruction error is proportional to $\sum_{i=1}^k\lambda_i$.
 On the other hand, the frames used in Theorem~\ref{thm5} are not uniform in general;
 this result in a codeword variance proportional to $\sum_{i=1}^k\lambda_i$ while having a
 fixed, minimum reconstruction error.

\subsection{Consistent Reconstruction}
\label{subsec:consist}

Linear reconstruction is not always the best one can estimate $\bm{x}$ from $\hat{\bm{y}}$.
Although {\it linear reconstruction} is more tractable, {\it consistent reconstruction} is known to give
significant improvement over  linear reconstruction in overcomplete expansions
\cite{thao1994consistent,thao1994reduction,goyal1998quantized}.
Asymptotically, the MSE is $O(r^{-2})$ for consistent reconstruction,
where $r=n/k$ is the frame redundancy \cite{thao1994reduction}.
As it can be seen from \eqref{eq:G6}, for linear reconstruction this is $O(r^{-1})$.
The improvement, in consistent reconstruction, is due to using deterministic properties of
quantization rather than considering quantization
as an independent noise as in \eqref{eq:Q}.

Although the MSE in consistent reconstruction is approximated by $cr^{-2}$,
where the constant $c$ depends on the source and quantization, this is verified only if
the oversampling ratio $r$ is very high \cite{goyal1998quantized}.
In some practical applications of frames, e.g., channel coding, this ratio cannot be
high, though. Particularly, in the context of interest, i.e., DSC, $r$
is limited to two \cite{vaezi2011DSC}. Besides, consistent reconstruction methods
 do not provide a guidance on how to design the frame, as they do not
 point out how to compute the constant $c$.
 More importantly, \eqref{eq:G6} proves to be predictive of the performance of consistent reconstruction \cite{goyal2001quantized};
 therefore, it can be convincingly used as a design criterion regardless of the reconstruction method.

\subsection{Reconstruction with Error and Erasure}
\label{subsec:consist}
In the context of channel coding, DFT codes are primarily used to provide robustness
against channel impairments which can be errors or erasures. Likewise, in DSC
these codes play the role of channel codes to combat the errors due to the virtual correlation channel \cite{vaezi2011DSC}.
Thus, it makes sense to evaluate the performance of these codes in the presence of error.
To this end, let $\bm{\hat y}  = G\bm{x}+ \bm{\eta}$ where $\bm{\eta}  = \bm{q}+ \bm{e}$.
Assuming that the quantization and channel errors are independent, we will have
\begin{equation}
\begin{aligned}
  \mathbb{E} \{ \bm{\eta}^T\bm{\eta} \}
  &=  \mathbb{E} \{ \bm{q}^T\bm{q} + \bm{q}^T\bm{e} + \bm{e}^T\bm{q} + \bm{e}^T\bm{e} \}  \\
&=  n \sigma_q^2 +\nu \sigma_e^2,
\label{eq:G8}
\end{aligned}
\end{equation}
where $\nu$ is the average number of errors in each codeword and $ \mathbb{E} \{\bm{e}^T\bm{e} \} \triangleq \nu \sigma_e^2 $.
Note that $ \mathbb{E} \{ \bm{e}^T\bm{q} \} = \mathbb{E} \{ \bm{q}^T\bm{e} \} = 0$, because $q$ and $e$ are
independent and $q$  has mean equal to zero.
Finally, following a similar analysis as in \eqref{eq:G6}, we obtain
 \begin{align}
\mathop{\mathrm{MSE_{q+e}}} = \frac{k}{n} \sigma_\eta^2 = \frac{k}{n}  \left( \sigma_q^2 +\frac{\nu}{n} \sigma_e^2 \right).
\label{eq:G9}
\end{align}

\noindent From \eqref{eq:G9} it is clear that reconstruction error has two
distinct parts caused by the quantization and channel errors. It also proves
that DFT codes decrease both channel and quantization errors
by a factor of frame redundancy $r=n/k$.
The above results is for the case when no error correction is done.
It is worth noting
that, even without correcting errors, the MSE can be smaller than quantization error.

As another extreme case, let us consider the case when
error localization is perfect, i.e., errors are in the {\it erasure}
form. Then, we remove the corrupted samples and
do reconstruction using the error-free samples. This approach does not
require error correction in order to reconstruct the message;
however, it is  shown to be equal to the coding theoretic approach \cite{rath2004frame}.
Let $\hat{\bm{y}}_R$ and $\bm{\eta}_R$ denote
remaining rows of $\hat{\bm{y}}$ and $\bm{\eta}$, respectively. Obviously, $\bm{\eta}_R$
includes only quantization error, hence we
represent $\bm{\eta}_R$ with $\bm{q}_R$.
 Also, let $F$ denote the rows of $G_{\mathrm{sys}}$
corresponding to $\bm{q}_R$.
Then, we can write
\begin{align}
\hat{\bm{y}}_R & = F\bm{x}+ \bm{q}_R,\\
\hat{\bm{x}}  &= F^\dagger \hat{\bm{y}}_R,
\label{eq:yhatR}
\end{align}
where $F^{\dagger}=(F^HF)^{-1}F^H$.
Thus, similar to \eqref{eq:G6} we will have
\begin{equation}
\begin{aligned}
\mathop{\mathrm{MSE_{q+\rho}}} &= \frac{1}{k} \mathbb{E}\{\|\bm{\hat x} -\bm{ x}\|^2\}
= \frac{1}{k}  \mathbb{E}\{\|F^{\dagger}\bm{q}_R\|^2\}  \\
&= \frac{1}{k}  \sigma_q^2 \tr\left(F^{\dagger H}F^{\dagger}\right)\\
&= \frac{1}{k}  \sigma_q^2 \tr\left(F^{H}F\right)^{-1}\\
&= \frac{1}{k} \sigma_q^2  \sum_{i=1}^{k}\frac{1}{\mu_i},
\label{eq:G10}
\end{aligned}
\end{equation}
where subscript $\rho$ denotes erasure and $\mu_{1} \geq \mu_{2}\geq \cdots \geq \mu_{k}>0$
represent the eigenvalues of $F^{H}F$. We assume at least $k$ samples are intact
which implies $\mu_{k}>0$.

One nice property of systematic frames is that reconstruction error cannot be
more than quantization error as long as systematic samples are intact.
This holds even if consecutive samples are erased. We know that consecutive erasures can
increase the MSE very fast (e.g., see \cite[Table I]{rath2004frame}).
This can be understood from \eqref{eq:G10} since $F$ contains $I_k$ as a subframe
and in the worst case we can use this subframe for reconstruction which leads to
$\mathop{\mathrm{MSE_{q+\rho}}} = \sigma_q^2 $. Adding any other row (sample) will decrease
the MSE. To show this, let $F^H=[I_k \, |\, E^H]$. Then,  $F^{H}F = I_k + E^{H}E$ and,
from  \eqref{eq:Weyl2}, for $i=j$, we get  $\mu_{i} \geq 1 + \xi_{k}$ for  $ i=1,\hdots,k$,
where $\xi_k$  is the smallest eigenvalue of $E^{H}E$.
Clearly, $\xi_k \geq 0$ since $E^{H}E$ is a positive semidefinite matrix. Further, at least
$\mu_{1}>0$ since otherwise $E$ must be zero. Hence,  $\sum_{i=1}^{k}\frac{1}{\mu_i}$
decreases by adding new rows.

Finally, with consistent reconstruction, we can further decrease the MSE. To do so,
we check if reconstructed values $\hat{x}_i$ for systematic samples in \eqref{eq:yhatR}  are consistent
with their values before reconstruction or not, i.e., for any systematic sample,
we must have $Q(\hat{x}_i)=Q(\hat{y}_{Ri})$. Otherwise, we  replace $\hat{x}_i$ with
\begin{align}
\hat{\hat{x}}_i  = Q(\hat{y}_{Ri}) - {\operatorname{sign}} (Q(\hat{y}_{Ri}) - \hat{x}_i)\frac{\Delta}{2}.
\label{eq:xhathat}
\end{align}

\section{Characterization of Systematic Frames }
\label{sec:class}

\subsection{The Best and Worst Systematic Frames}
\label{sec:class1}

As we discussed in Section~\ref{sec:eig}, the optimal $G_{\mathrm{sys}}$ is achieved
from the optimization problem \eqref{eq:Omin}.
Similarly, to find the worst $G_{\mathrm{sys}}$, we can {\it maximize} \eqref{eq:Omin} instead of minimizing it.
The optimal eigenvalues are known to be $\lambda_i=1, 1 \leq i\leq k$.
But, how can we find the corresponding $G_{\mathrm{sys}}$, or $G_{k}$ equivalently?
More importantly, if a $G_{k}$ with  $\lambda_i=1$ does not exist,
is there any suggestion for the best matrix?

We approach this problem by studying another optimization problem. To this end, we first prove
the following theorem for the eigenvalues of $G_{k}G_{k}^H$.
\begin{thm}
\label{thm6}
Let $\{\lambda_i\}_{i=1}^k$  be the eigenvalues of $G_{k}G_{k}^H$, where $G_{k}$ includes $k$ arbitrary rows of $G$, then we have
\begin{align}
\underset{\lambda_i}{\operatorname{argmin}}  \sum_{i=1}^k\frac{1}{\lambda_i} = \underset{\lambda_i}{\operatorname{argmax}} \prod_{i=1}^k \lambda_i.
 \label{eq:Oeq}
 \end{align}
\end{thm}

\begin{proof}
See Section~\ref{sec:App2}.
\end{proof}
\noindent Now, in view of Theorem~\eqref{thm6},  the optimal arguments
of the optimization problem in \eqref{eq:Omin} are equal to those of
\begin{equation}
\begin{aligned}
& \underset{\lambda_i}{\text{maximize}}
& & \prod_{i=1}^k \lambda_i\\
& \text{s.t.}
& & \sum_{i=1}^k\lambda_i=k, \,\, \lambda_i>0,
\end{aligned}
\label{eq:Omax}
\end{equation}
in which $\{\lambda_i\}_{i=1}^k$ are the eigenvalues of $G_{k}G_{k}^H$ (or $V_{k}^HV_{k}$).
By using the Lagrangian method, one can check that \eqref{eq:Omax} has the maximum of 1 and infimum of 0. Then,
considering that
\begin{equation}
\begin{aligned}
\prod_{i=1}^k \lambda_i = \det (V_{k}^HV_{k})=  \det (G_{k}G_{k}^H),
\end{aligned}
\label{eq:det}
\end{equation}
we conclude that the ``best'' submatrix is the one with the largest
determinant (possibly 1) and the ``worst'' submatrix is the one with smallest determinant.

Next, we evaluate the determinant of $V_{k}^HV_{k}$ so as to
find the matrices corresponding to the extreme cases.
To this end, we first evaluate the determinate of $WW^H$ where $W$ is
the Vandermonde matrix with unit complex entries as defined in
\eqref{eq:Vand2}. From \eqref{eq:detV0} we can write
 \begin{equation}
\begin{aligned}
 \det(WW^H) &= \frac{1}{n^n}\prod_{\substack 1\leq p<q\leq n}{|e^{i\theta_p}-e^{i\theta_q}|^2} \\
  &= \frac{1}{n^n}\prod_{\substack 1\leq p<q\leq n}{4 \sin^2 \frac{\pi}{n}(q - p)} \\
  &= \frac{2^{n(n-1)}}{n^n}\prod_{\substack r=1}^{n-1}{\left( \sin^2 \frac{\pi}{n}r\right)^{n-r}},
 \label{eq:detV}
\end{aligned}
\end{equation}
in which $\theta_x = \frac{2\pi}{n}(x-1),  r=q - p$, and $n(n-1)/2$ is the total number of terms that
satisfy $1\leq p<q\leq n$.
But, we see that $W$ is a DFT matrix, and thus,
its determinant must be 1. Therefore, we have
 \begin{align}
 \prod_{\substack r=1}^{n-1}{\left( \sin^2 \frac{\pi}{n}r\right)^{n-r}}=\frac{n^n}{2^{n(n-1)}}.
 \label{eq:formula}
 \end{align}

 The above analysis helps us evaluate the determinant of $V_{k}$  or  $G_{k}$, defined in \eqref{eq:Gsys3}.
Let $\mathcal I_{r_k}=\{i_{r_1}, i_{r_2}, \hdots, i_{r_k}\}$ be those rows of $G$ used to
build $G_{k}$.  Also, without loss of generality, assume $i_{r_1}< i_{r_2}< \cdots < i_{r_k}$. Clearly, $i_{r_1}\geq 1, i_{r_k}\leq n$, and
we obtain
\begin{equation}
\begin{aligned}
 \det(V_kV_k^H) &= \frac{1}{k^k}\prod_{\substack {1\leq p<q\leq n \\ p, q \in \mathcal I_{r_k}}}{|e^{i\theta_p}-e^{i\theta_q}|^2} \\
  &= \frac{1}{k^k}\prod_{\substack {1\leq p<q\leq n \\ p, q \in \mathcal I_{r_k}}}{4 \sin^2 \frac{\pi}{n}(q - p)}.
 \label{eq:detVk}
 \end{aligned}
\end{equation}
Then, since $\sin \frac{\pi}{n}u = \sin \frac{\pi}{n}(n-u) $, one can see that this determinant depends on the
circular distance between rows in $\mathcal I_{r_k}$. For a matrix with $n$ rows, we define the {\it circular distance} between rows $p$ and $q$
as $\min{\{|q - p|, n-|q - p|\}}$. In this sense, for example, the distance
between rows $1$ and  $n$ is one.
Now, it is reasonable to believe that \eqref{eq:detVk} is minimized when the selected rows are {\it circularly successive}.\footnote{
A set of J rows $\{i_{r_1}< i_{r_2}< \hdots< i_{r_J}\}$ of a matrix
are successive if they are one after the other, i.e.,  $i_{r_j}=i_{r_{j-1}+1}$.
A set of rows are circularly successive if they or their complement set of rows are successive, where
the complement of a set of rows includes all rows except that set of rows.
}
Note that $\sin u$ is strictly increasing for  $u \in [0, \pi/2]$, and the circular distance cannot be greater than $n/2$ in this problem.

In such circumstances where all rows in ${\mathcal I_{r_k}}$ are (circularly) successive, \eqref{eq:detVk} is minimal and reduces to
 \begin{align}
 \det(V_kV_k^H) = \frac{2^{k(k-1)}}{k^k}\prod_{\substack r=1}^{k-1}{\left( \sin^2 \frac{\pi}{n}r\right)^{k-r}}.
 \label{eq:detVk2}
 \end{align}
The other extreme case comes up when $n=Mk$ ($M$ is a positive integer) provided that $G_{k}$ consists of every $M$th row of $G$.
In such a case, \eqref{eq:detVk} simplifies to $1$, because

\begin{equation}
\begin{aligned}
 \det(V_kV_k^H) &=  \frac{2^{k(k-1)}}{k^k}\prod_{\substack r=1}^{k-1}{\left( \sin^2 \frac{\pi}{n}Mr\right)^{k-r}}  \\
  &=  \frac{2^{k(k-1)}}{k^k}\prod_{\substack r=1}^{k-1}{\left( \sin^2 \frac{\pi}{k}r\right)^{k-r}}\\& =1,
 \label{eq:detVm}
\end{aligned}
\end{equation}
where the last step is because of \eqref{eq:formula}.
Recall that this gives the best $V_{k}$ (and equivalently $G_{k}$), in light of \eqref{eq:Omax}.
For such a $G_{k}$, it is easy to see that $G_{\mathrm{sys}}$ stands for
a ``tight'' systematic frame and minimizes the MSE for a given number of quantization levels.
Effectively, such a frame is performing {\it integer oversampling}.
There are $M$ such frames; they all have the same spectrum, though.

Recall that, from  \eqref{eq:Oeq}--\eqref{eq:det}
and Theorem~\ref{thm2}, $\det(V_kV_k^H) < 1$ for $n \neq Mk$. For such an $(n,k)$
frame, the systematic rows cannot be equally
spaced in the corresponding systematic frame; instead, we may explore a systematic frame
in which the circular distance between successive systematic samples is
as evenly as possible. Then, the circular distance between each successive systematic rows
is either $\lfloor n/k \rfloor $ or $\lceil n/k \rceil $. More precisely, if $l$ and $m$, respectively,
represent the number of systematic rows with circular distance equal to
$\lceil n/k \rceil $ and $\lfloor n/k \rfloor $, they must satisfy 
\begin{align}
\begin{cases}
l+m=k,\\
l\lceil \frac{n}{k} \rceil +m\lfloor \frac{n}{k} \rfloor =n.
\label{consts}
\end{cases}
\end{align}

\noindent In the following theorem, we prove that the best performance is
achieved when the systematic rows are as equally spaced as possible, i.e., when
 \eqref{consts} is satisfied.

\begin{thm}
When encoding with an $(n,k)$ systematic DFT frame in \eqref{eq:Gsys2} and decoding with linear reconstruction,
for the noise model \eqref{eq:Q} and given a same number of quantization levels,
 the MSE is minimum when there are $l= n- \lfloor n/k \rfloor k$ systematic rows with successive circular distance $\lceil n/k \rceil $
 and the remaining $m= k -l$ systematic rows have a successive circular distance equal to $\lfloor n/k \rfloor $.
\label{thm7}
\end{thm}

\begin{proof}
See Appendix~\ref{sec:App4}.
\end{proof}

Effectively, the above theorem is generalizing  Theorem~\ref{thm5}.
Note that when $n = Mk$, $\lfloor n/k \rfloor = \lceil n/k \rceil=M$ and there exist $k$
systematic rows with equal distance; in this case, Theorem~\ref{thm7} reduces to
Theorem~\ref{thm5} and the  corresponding systematic frame is tight.
The optimality of this case was proved in \eqref{eq:detVm}. When $n \ne Mk$,
we cannot have a systematic frame with equally spaced systematic rows; however,
the best performance is still achieved when the circular distance between
the systematic (parity) rows is as evenly as possible, as detailed above.
Note that in either case $d_{\mathrm{min}}$, the minimum distance between the systematic
rows, is $ \lfloor n/k \rfloor$. This is a necessary condition for an
optimal systematic frame, as shown in the proof of Theorem~\ref{thm7}.
Further, to satisfy Theorem~\ref{thm7}, the minimum distance between
the parity rows must be $\bar d_{\mathrm{min}}= \lfloor n/(n-k) \rfloor$.

\begin{table*}[t]
\caption{Eigenvalues structure for two systematic DFT frames with different codeword patterns.
A ``$\times$'' and ``$-$'' respectively represent data (systematic) and parity samples.} \label{table1}
\centering
\scalebox{.99}{
\begin{tabular}{llcccccc}
\toprule %
\multicolumn {2}{c}{$\mathop{\mathrm{Code \quad\quad Codeword }} \qquad\;  $} & $\lambda_{\min}$  & $\lambda_{\max}$ & $\sum_{i=1}^k 1/\lambda_i$ & $\prod_{i=1}^k\lambda_i$ \\
\multicolumn {2}{c}{$\mathop{\mathrm{\qquad\quad\quad\quad  patern}} \qquad\;  $} &   &  &  &  \\ \toprule \toprule %
& $\times\times\times  - - -$  & $0.0572$ & $1.9428$ & $19$ & $0.1111$ &   \\
\cmidrule (r){3-7}
\multirow {2}*{$(6, 3)$}
& $\times\times -\times  - -$ & $0.2546$ & $1.7454$ & $5.5$ & $0.4444$ &   \\
\cmidrule (r){3-7}
& $\times\times --\times   -$ & $0.2546$ & $1.7454$ & $5.5$ & $0.4444$ &   \\
\cmidrule (r){3-7}
& $\times -\times -\times-$ & $1$ & $1$ & $3$ & $1$ &    \\
\midrule
& $\times\times\times \times\times - -$  & $0.0396$ & $1.4$ & $28.70$ & $0.0827$ &   \\
\cmidrule (r){3-7}
\multirow {2}*{$(7, 5)$}
& $\times\times \times\times -\times - $ & $0.1506$ & $1.4$ & $10.32$ & $0.2684$ &   \\
\cmidrule (r){3-7}
& $\times\times -\times\times -\times $ & $0.3110$ & $1.4$ & $7.40$ & $0.4173$ &   \\
\cmidrule (r){3-7}
& $\times -\times \times\times -\times $ & $0.3110$ & $1.4$ & $7.40$ & $0.4173$ &   \\
\midrule
& $\times\times\times\times\times - - - - -$  & $0.0011$ & $1.9989$ & $908.21$ & $4.46 \times 10^{-4}$ &   \\
\cmidrule (r){3-7}
& $\times\times\times\times -\times  - - - -$ & $0.0041$ & $1.9959$ & $249.94$ & $0.0047$ &   \\
\cmidrule (r){3-7}
& $\times\times\times\times - -\times  - - -$  & $0.0110$ & $1.9890$ & $96.09$ & $0.0122 $ &   \\
\cmidrule (r){3-7}
& $\times\times\times - \times - - - \times -$  & $ 0.0202$ & $1.9798$ & $53$ & $0.0400$ &   \\
\cmidrule (r){3-7}
& $\times\times\times - - \times\times - - -$  & $0.0496$ & $1.9504$ & $25.64$ & $0.0489$ &   \\
\cmidrule (r){3-7}
& $\times\times\times - \times -  \times - - -$  & $ 0.0310$ & $1.9690$ & $35.73$ & $0.0611$ &   \\
\cmidrule (r){3-7}
\multirow {2}*{$(10, 5)$}
& $\times \times \times - - \times- -\times-$  & $0.0512$ & $1.9488$ & $23.41$ & $0.0838$ &   \\
\cmidrule (r){3-7}
& $\times \times \times - - \times-\times- -$  & $0.0835$ & $1.9165$ & $16$ & $0.1280$ &   \\
\cmidrule (r){3-7}
& $\times\times- \times \times - - \times - -$  & $0.1056$ & $1.8944$ & $13.79$ & $0.1436$ &   \\
\cmidrule (r){3-7}
& $\times \times - - \times \times - -\times -$ & $0.2497$  & $1.7503$ & $9.56$ & $0.2193$ &    \\
\cmidrule (r){3-7}
& $\times \times - \times -\times - -\times -$ & $0.1902$  & $1.8098$ & $8.86$ & $0.3351$ &    \\
\cmidrule (r){3-7}
& $\times - \times - \times -\times - -\times $ & $0.2377$  & $1.7623$ & $7.77$ & $ 0.4189$ &    \\
\cmidrule (r){3-7}
& $\times -\times -\times-\times -\times  -$ & $1$ & $1$ & $5$ & $1$ &    \\
\bottomrule
\end{tabular}
}

\end{table*}

\subsection{Numerical Examples}
\label{sec:numerical}

Numerical calculations confirm that ``evenly" spaced data samples gives rise to
systematic frames with the best performance.
When a systematic frame is doing integer oversampling, 
we end up with tight systematic frames. The first  and last codes in Table~\ref{table1}
are examples of this case. When $n\neq Mk$, data samples cannot be equally
spaced; however, as it can be seen from the second code in Table~\ref{table1},
still the best performance is achieved when they are as equally spaced as possible.
In this table, ``$\times$'s'' and ``$-$'s'' represent data and parity samples, respectively.
Moreover, we observe that
circularly shifted codeword patterns behave the same
(e.g., in the $(7, 5)$ code, frames with pattern $\times -\times \times\times -\times $ and $\times\times -\times\times -\times $
have the same performance).
 Also, reversal of a codeword pattern yields a codeword with the same performance
(e.g., $\times\times -\times - -$ is shifted version of reversed $\times\times --\times   -$ in the $(6,3)$ code).
These properties hold in general, as stated below.

\begin{property}
Circular shift of ${\mathcal I}_{r_k}$, the systematic rows of a systematic frame with analysis frame $G_{\mathrm{sys}}$, does not
change the spectrum of $G_{\mathrm{sys}}^HG_{\mathrm{sys}}$.
\label{property1}
\end{property}

\begin{property}
Reversal of ${\mathcal I}_{r_k}$ yields a systematic frame with the same spectral properties.
\label{property2}
\end{property}

\begin{proof}
From \eqref{eq:Gsys2} we obtain
\begin{align}
\lambda_i(G_{\mathrm{sys}}^HG_{\mathrm{sys}})= \frac{n/k}{\lambda_i(G_{k}G_{k}^H)},
\end{align}
for $i=1,\hdots,k$. But $G_{k}G_{k}^H$ is invariant to the circular shift of rows of $G$ that make ${\mathcal I}_{r_k}$, as long as all rows are shifted
the same amount in the same direction. This can be seen from the proof of Lemma~\ref{lem1} in \eqref{eq:Toep}
by defining $r'=r+c$ where $r'$ represent the shifted rows by a constant $c$ and $r \in {\mathcal I}_{r_k}$. This proves
Property~\ref{property1}. Likewise,
let $r''=n+1-r$ be the reversed row indices. Again, from \eqref{eq:Toep}, it is clear that
Property~\ref{property2} holds.
\end{proof}
\noindent These properties together show that the frame operators
of systematic frames ($G_{\mathrm{sys}}^HG_{\mathrm{sys}}$), in which the ``relative" circular distance
among the systematic rows are the same, inherit the same spectrum and thus show the same performance.

\subsection{Number of Systematic Frames}
\label{sec:number}
The number of systematic frames is obviously finite but their performance
depends on the position of the systematic rows, or equivalently,
the position of data (or parity) samples in the associated codewords, and can be the same
for different systematic frames. In what
follows, we derive an  upper and  lower bound on the number of systematic frames
with different spectrum. In other words, we categorize these frames
based on their performance. To this end, we observe that the problem of
finding $k \times k$ submatrices of an $n \times k$ matrix
can be viewed as finding different $k$-subsets of a set with $n$ elements.
This is given by the binomial coefficient $\binom{n}{k}$ and is also equivalent
to the number of systematic frames. As stated earlier in Property~\ref{property1},
circular shift of a codeword pattern does not change its spectrum, and so its performance.
We define a {\it coset} as square submatrices that result in a same performance.
Each coset has at least $n$ elements ($k$-subsets), as shown in Table~\ref{table2}.
To find these elements, it suffices to circularly shift a subset $n$ times.
Equivalently, for a given $k$-subset, we simply add up $1$ to each element of a subset.
Note that, the subsets elements are $k$ row indices of $G_{n \times k}$ and thus
cannot be greater than $n$. Therefore, once a shifted index $x$ becomes greater than $n$,
we replace it with $\llangle x \rrangle_n$
where $\llangle x \rrangle_n \triangleq x-dn \;\text{if} \; dn+1 \leq x \leq dn+n, d\in \mathbb Z.$
Obviously, each coset has at least $n$ subsets since $n-1$ circular shifts of
a given subset are distinct; all these subsets have the same relative distance, though.
This can be seen in
Table~\ref{table2}. Thus, it is clear that the number of cosets is the bounded by
\begin{align}
n_c \leq u =\frac{1}{n}\binom{n}{k}.
\label{eq:binom}
\end{align}

\noindent Let $\mathcal I_{r_k}^r$ denote the reversal of
$\mathcal I_{r_k}=\{i_{r_1}, i_{r_2}, \hdots, i_{r_k}\}$ where
\begin{align}
\mathcal I_{r_k}^r \triangleq \llangle n+1-\mathcal I_{r_k} \rrangle_n.
\label{eq:rev}
\end{align}
This operation is performed on every element of $\mathcal I_{r_k}$.
One can see that reversal of a subset
does not change its distance and spectrum, owing to Property~\ref{property2}. This can reduce
the number of cosets.
For example, in Table~\ref{table2}, the reversal of $\{1,2,4\}$, which is the {\it coset leader} in $\mathrm{C}_2$,
is $\{7,6,4\}$ which belongs to $\mathrm{C}_5$. This indicates $\mathrm{C}_2$ and $\mathrm{C}_5$ are essentially one coset.
 The bound in \eqref{eq:binom}
is tight if and only if there are $u$ self-reversal cosets.
Trivial examples of such a code are achieved when $k=n-1$ or $k=1$.
A self-reversal coset is a coset that the reversal of its elements belong to itself, e.g.,
$\mathrm{C}_1$, $\mathrm{C}_3$, and $\mathrm{C}_4$ in Table~\ref{table2}.

On the other hand, $n_c \geq u/2$ is a lower bound because there
cannot be more than one reversal for a given coset.
It can be further seen that the coset with smallest weight ($\mathrm{C}_1$) is always self-reverse,
i.e., the reversal of each element of $\mathrm{C}_1$ is its own element for any $(n,k)$ frame.
This implies that the lower bound is not achievable. Therefore,
\begin{align}
\frac{1}{2n}\binom{n}{k} < n_c  \leq \frac{1}{n}\binom{n}{k}.
\label{eq:binom2}
\end{align}
One can check that the first two frames in Table~\ref{table1} reach the upper bound
 $\lfloor \frac{1}{n}\binom{n}{k} \rfloor$ whereas the
third one satisfies the lower bound $\lceil \frac{1}{2n}\binom{n}{k} \rceil$.

\begin{table}[t]
\caption{Different cosets of $(7,3)$ DFT frame and their corresponding relative
distances and spectrums. The Coset leaders are in boldface.} \label{table2}
\centering
\scalebox{.99}{
\begin{tabular}{llcccccc}
\toprule %
\multicolumn {1}{c}{$$} & \quad $\mathrm{C}_1$  & $\mathrm{C}_2$ & $\mathrm{C}_3$ & $\mathrm{C}_4$ & $\mathrm{C}_5$ \\
\toprule \toprule %
   \multirow {1}*{$\mathrm{Leader}$}
 &  $\bold{1\;\; 2\;\; 3}$ & $\bold{1\;\; 2\;\; 4}$ & $\bold{1\;\; 2\;\; 5}$ & $\bold{1\;\; 3\;\; 5}$ & $\bold{1\;\; 3\;\; 4}$ &  \\
  & $2\;\; 3\;\; 4$ & $2\;\; 3\;\; 5$ & $2\;\; 3\;\; 6$ & $2\;\; 4\;\; 6$ &  $2\;\; 4\;\; 5$ \\
  & $3\;\; 4\;\; 5$ & $3\;\; 4\;\; 6$ & $3\;\; 4\;\; 7$ & $3\;\; 5\;\; 7$ &  $3\;\; 5\;\; 6$ \\
  & $4\;\; 5\;\; 6$ & $4\;\; 5\;\; 7$ & $1\;\; 4\;\; 5$ & $1\;\; 4\;\; 6$ &  $4\;\; 6\;\; 7$ \\
 & $5\;\; 6\;\; 7$ & $1\;\; 5\;\; 6$ & $2\;\; 5\;\; 6$ & $2\;\; 5\;\; 7$ &  $1\;\; 5\;\; 7$ \\
   & $1\;\; 6\;\; 7$ & $2\;\; 6\;\; 7$ & $3\;\; 6\;\; 7$ & $1\;\; 3\;\; 6$ &  $1\;\; 2\;\; 6$ \\
  & $1\;\; 2\;\; 7$ & $1\;\; 3\;\; 7$ & $1\;\; 4\;\; 7$ & $2\;\; 4\;\; 7$ &  $2\;\; 3\;\; 7$ \\
  \cmidrule (r){1-7}
  \multirow {1}*{$\mathrm{Distance}$}
 & $1\;\; 1\;\; 2$ & $1\;\; 2\;\; 3$ & $1\;\; 3\;\; 3$ & $2\;\; 2\;\; 3$ &  $1\;\; 3\;\; 2$ \\
  \cmidrule (r){1-7}
  \multirow {1}*{$\mathrm{Weight}$}
 & \quad $4$ & $6$ & $7$ & $7$ &  $6$ \\
 \cmidrule (r){1-7}
  \multicolumn {1}{c}{$\lambda_1$} & $2.1558$  & $ 1.7539$ & $1.9066$ & $1.2673$ & $1.7539$ \\
  \multicolumn {1}{c}{$\lambda_2$} & $0.8150$  & $1.1133$ & $0.8424$ & $1.1601$ & $1.1133$ \\
  \multicolumn {1}{c}{$\lambda_3$} & $0.0292$  & $0.1328$ & $0.2510$ & $0.5726$ & $ 0.1328$ \\
\bottomrule
\end{tabular}
}
\end{table}

\section{Conclusions}
\label{sec:con}
We have introduced the application, proposed the construction method, and analyzed the performance
of systematic DFT frames in this paper. Numerous  systematic DFT frames
can be made out of one DFT frame;
the performance of these frames differs depending on
the relative position of the systematic and parity samples in the codeword.
We proved that evenly spaced systematic (or parity)
samples result in the minimum mean-squared reconstruction error,
whereas the worst performance is expected when the parity samples are circularly consecutive.
Further, we found the conditions for which a systematic DFT frame 	
can be tight, too. A tight systematic DFT frame can be realized only
if the frame is performing integer oversampling and systematic samples are circularly equally spaced.
Finally, for each DFT frame, we classified systematic DFT frame based on their performance.

It would be interesting to extend this work to oversampled DFT filter banks, an infinite-dimension of DFT frames,
since oversampled filter banks can be used for
error correction \cite{labeau2005oversampled}.


\section*{Acknowledgement}
The authors wish to thank the reviewers
of the conference and journal versions of this work for several valuable comments. We also thank
Sina Hamidi Ghalehjegh and Mohsen Akbari for fruitful discussions and their comments.

\section{Appendix}

\subsection{Proof of Theorem~\ref{thm2}}
\label{sec:App1}
\begin{proof}
Let $n=Mk+l, 0<l<k$, then $G$ can be partitioned  as $G=[G_{k}^{H} \,|\, G_{k}^{1H} \,|\,\cdots \,
|\,G_{k}^{(M-1)H}\,|\,G_{k\times l}^{MH}]^H$.
In general, $G_{k},G^{1}_{k}, \hdots ,G^{M-1}_{k}$ and $G_{k\times l}^{M}$ include arbitrary rows of $G$, hence they have
different spectrums, i.e., different sets of eigenvalues.
%
Suppose, for the purpose of contradiction, that $\lambda_{k}(G_{k}^HG_{k})=1$; this can occur only if $G_{k}$ consist of the rows of $G$
such that the distance between each two successive rows is at least $M$.\footnote{$\lambda_{k}(G_{k}^HG_{k})=1$ is the
optimal solution for \eqref{eq:Omin} and necessitate $d_{\mathrm{min}}=M$, as discussed in Theorem~\ref{thm7}.}
Such an arrangement guarantees the existence of $G_{k}^{1}, \hdots ,G_{k}^{M-1}$
so that $G_{k}^{mH}G_{k}^{m}$, for any $1 \leq m\leq M-1$, has the same spectrum as $G_{k}^HG_{k}$.
To find the row indices corresponding to $G^{m}_{k}$,
we can simply add $m$ to each row index of $G_{k}$.
Then, to show these matrices have the same spectrum, we use  Lemma~\ref{lem0}.
Given a $G_{k}$, one can verify that $(G^{m}_{k})_{i,j} =  e^{j\frac{2\pi m}{n}} (G_{k})_{i,j}$ and thus
 $(G_{k}^{m})_{i,j}^H =  e^{-j\frac{2\pi m}{n}}(G_{k})_{i,j}^H $. Therefore,
 $G_{k}^{mH} G^{m}_{k}$ and $G_{k}^{H} G_{k}$
 have the same spectrum for any $1 \leq m\leq M-1$.
 Next, we see that $G^HG= A+B$ in which $A= G^{H}_{k} G_{k}+ \cdots +G^{(M-1)H}_{k} G^{M-1}_{k}$ and
$B=G^{MH}_{k \times l}G^M_{k \times l}$. Then, in consideration of the above discussion, $\lambda_{i}(A)= M\lambda_{i}(G_{k}^H G_{k})$ for any $1\leq i \leq k$.
Hence, from \eqref{eq:Weyl2}, for $i=1, j=k$, we will have
  \begin{equation}
\begin{aligned}
  \lambda_{k}(A)+\lambda_{1}(B)&\leq   \lambda_{1}(A+B)\\
  \Leftrightarrow  M\lambda_{k}(G_{k}^H G_{k}) &\leq\frac{n}{k} -\lambda_{1}(B) \\
    \Leftrightarrow  \lambda_{k}(G_{k}^H G_{k}) &\leq    \frac{\frac{n}{k}-1}{M} = \frac{\frac{n}{k}-1}{\lfloor\frac{n}{k} \rfloor }<1,
\end{aligned}
\label{eq:thm2}
 \end{equation}
where the last line follows using $\lambda_{1}(B)\geq 1$ from Theorem~\ref{thm1}.
But this is contradicting our assumption $\lambda_{k}(G_{k}^HG_{k})=1$, and thus
completes the proof that, for $n\ne Mk$, the largest possible $\lambda_{k}(G_{k}^H G_{k})$ is strictly
less than 1, for any $G_{k}$.\footnote{Note that when $n=Mk$, $B$ is an empty matrix and we must plug $\lambda_{1}(B)=0 $
into \eqref{eq:thm2} which result in $\lambda_{k}(G_{k}^H G_{k}) \leq 1$
and does not guarantee a bound strictly less than 1.}

The proof of the other bound ($\lambda_{1}(G_{k}^H G_{k})> 1$) is then immediate because
$\sum_{\substack i=1}^{k}{\lambda_{i}(G_{k}^H G_{k})} =\sum_{\substack i=1}^{k}a_{ii}=k$.
\end{proof}

\subsection{Proof of Theorem~\ref{thm6}}
\label{sec:App2}
\begin{proof}
Let $\{\lambda_i\}_{i=1}^k$ be the eigenvalues of $G_{k}G_{k}^H$.
From Lemma~\ref{lem1}, we know that $\sum_{\substack i=1}^{k}{\lambda_{i}(G_{k}^H G_{k})} =k$.
Then, subject to this constraint, by using the Lagrangian method \cite{boyd2004convex}, its is straightforward
to see that the optimal values of the optimization problems in both sides of \eqref{eq:Oeq} are $\lambda_i=1,  i=1,\hdots,k$.

\end{proof}

\subsection{Proof of Theorem~\ref{thm7}}
\label{sec:App4}
\begin{proof}
Consider an $(n,k)$ DFT frame, let $M=\lfloor n/k \rfloor $, and assume that all rows in
${\mathcal I_{r_k}}$, except the first and last rows, are equally spaced with distance $M$
(without loss of generality, we assume $i_{r_1}=1$, then $i_{r_j}=(j-1)M+1$, $j \le k$).
Hence $d_{\mathrm{min}}=M$, where the minimum distance $d_{\mathrm{min}}$ is defined as the smallest circular distance among the selected rows.
 In such a setting, from \eqref{eq:detVk} and similar to \eqref{eq:detVm}, we can write
 \begin{align}
 \det(V_kV_k^H) = \frac{2^{k(k-1)}}{k^k}\prod_{\substack r=1}^{k-1}{\left( \sin^2 \frac{\pi}{n}Mr\right)^{k-r}}.
 \label{eq:detVk2}
 \end{align}

We prove that, in view of \eqref{eq:Omax}, the systematic frame corresponding
to the above arrangement has better performance than any other arrangement
in which $d_{\mathrm{min}}$ among the systematic rows is less than $M$.
To this end, we first assume that all selected rows in $\mathcal I_{r_k}$ remain the same
except one row which is shifted one unit in a way that $d_{\mathrm{min}}$ decreases.
For example, without loss of generality, consider $\mathcal I'_{r_k}$
for which $i'_{r_1}=2,$ $ i'_{r_j}=i_{r_j}, 1 < j \le k$; hence $d_{\mathrm{min}}=M-1$.
Then, from \eqref{eq:detVk}, we  obtain
 \begin{align}
\frac{ \det(V_kV_k^H)|_{\mathcal I'_{r_k}}}{ \det(V_kV_k^H)|_{\mathcal I_{r_k}}}= \frac{\prod_{\substack r=1}^{k-1}{ \sin^2 \frac{\pi}{n}(Mr -1)}}{\prod_{\substack r=1}^{k-1}{ \sin^2 \frac{\pi}{n}Mr}} < 1.
 \label{eq:ratio}
 \end{align}

\noindent To prove the inequality, equivalently, we show that
 \begin{align}
  \frac{\sin \frac{(M -1)\pi}{n}\sin \frac{(2M -1)\pi}{n} \cdots \sin \frac{((k-1)M -1)\pi}{n}}{\sin \frac{M\pi}{n}\sin \frac{2M\pi}{n} \cdots \sin \frac{(k-1)M\pi}{n} } <1.
 \label{eq:ratio1}
 \end{align}

\noindent We break up this inequality into $\lfloor k/2 \rfloor $ inequalities, each of which strictly less than one.
First, consider the first and last terms in the numerator and denominator. We can write
 \begin{align}
  \frac{\sin \frac{(M -1)\pi}{n} \sin \frac{((k-1)M -1)\pi}{n}}{\sin \frac{M\pi}{n} \sin \frac{(k-1)M\pi}{n} }
  &=  \frac{\cos \frac{(k-2)M\pi}{n} - \cos \frac{(kM-2)\pi}{n}}{\cos \frac{(k-2)M\pi}{n} - \cos \frac{kM\pi}{n}}\nonumber \\
  & < 1,
 \label{eq:ratio2}
 \end{align}
where the inequality follows since $\cos \frac{(kM-2)\pi}{n} >  \cos \frac{kM\pi}{n}$, as $ \frac{kM}{n}\pi \le \pi $.
Likewise, for the second and penultimate terms we have
 \begin{align}
  \frac{\sin \frac{(2M -1)\pi}{n} \sin \frac{((k-2)M -1)\pi}{n}}{\sin \frac{2M\pi}{n} \sin \frac{(k-2)M\pi}{n} }
  &=  \frac{\cos \frac{(k-4)M\pi}{n} - \cos \frac{(kM-2)\pi}{n}}{\cos \frac{(k-4)M\pi}{n} - \cos \frac{kM\pi}{n}}\nonumber \\
  & < 1.
 \label{eq:ratio2}
 \end{align}
 A similar reasoning can be used for other terms that are equally spaced from the two ends.

Clearly, the same argument is valid when $2< i'_{r_1}<M$  and the other rows are the same, i.e., $i'_{r_j}=i_{r_j}, 1 < j \le k$
and $d_{\mathrm{min}}=M-i'_{r_1}$. Moreover, when more than one row index is changed, in a way that two or more selected rows have a distance less than $M$,
the above argument is valid and we can show that new determinant is even less than the case with one changed index.
In fact, in such a case, it is easier to compare the new one with its parent; i.e.,
to compare the case with two changes with the case with one change.      As a result, we can see that
any combination of rows with $d_{\mathrm{min}}<M$ performs worse than the case with $d_{\mathrm{min}}=M$, on account of (37);  that is, $d_{\mathrm{min}}=M$
is necessary condition for optimality. In other words, that optimal systematic frame must satisfy $d_{\mathrm{min}}=M$.

Next, we show that among systematic frames with $d_{\mathrm{min}}=M$ the one that satisfies \eqref{consts} is the best.
That is, the optimal systematic frame has
$l= n- \lfloor n/k \rfloor k$ systematic rows with successive circular distance of $\lceil n/k \rceil $
 and $m= k -l$ systematic rows with successive circular distance of $\lfloor n/k \rfloor$. To
prove this, again we compare $\det(V_kV_k^H)$ in  \eqref{eq:detVk} for this case and the other cases
with $d_{\mathrm{min}}=M$. The arguments are very similar to what we used above.
Before moving on, we should mention that for $l\in\{0,1,k-1\}$ the proof in the first part is sufficient.

 Let $\mathcal I^o_{r_k}$ denote the set of rows satisfying the constraints in \eqref{consts};
 obviously, $d_{\mathrm{min}}=M$.  We claim that any other selection of systematic rows, for which $d_{\mathrm{min}}$ is $M$,
 results in a smaller $\det(V_kV_k^H)$; that is, $\det(V_kV_k^H)|_{\mathcal I_{r_k}} < \det(V_kV_k^H)|_{\mathcal I^o_{r_k}}$.
 Let us evaluate the case where
 only the row index for one of those $l$ rows varies, provided that $d_{\mathrm{min}}=M$ is kept.\footnote{ Note that, with this shift of row,
 we are looking for an arrangement of a systematic frame that does not satisfy \eqref{consts}; otherwise, $\det(V_kV_k^H)$ will not
 vary, as the frame properties has not changed essentially. More specifically, a new, different arrangement will introduce a new distance equal to $\lceil n/k \rceil+1 $.}
 We then have
 \begin{align}
\frac{ \det(V_kV_k^H)|_{\mathcal I_{r_k}}}{ \det(V_kV_k^H)|_{\mathcal I^o_{r_k}}}= \frac{\prod_{\substack r=1}^{k-1}{ \sin^2 \frac{\pi}{n}Mr}}{\prod_{\substack r=1}^{k-1}{ \sin^2 \frac{\pi}{n}(Mr +1)}} < 1.
 \label{eq:ratio3}
 \end{align}
Again it suffice to prove that
 \begin{align}
  \frac{\sin \frac{M\pi}{n}\sin \frac{2M\pi}{n} \cdots \sin \frac{(k-1)M\pi}{n} } {\sin \frac{(M +1)\pi}{n}\sin \frac{(2M +1)\pi}{n} \cdots \sin \frac{((k-1)M +1)\pi}{n}} <1,
 \label{eq:ratio4}
 \end{align}
and this can be done by the same divide and conquer approach, used in the first part of this proof. For instance,
for the first and last terms in the numerator and denominator we have
\begin{align}
  \frac{\sin \frac{M\pi}{n} \sin \frac{(k-1)M\pi}{n} }{\sin \frac{(M +1)\pi}{n} \sin \frac{((k-1)M +1)\pi}{n}}
  &=  \frac{\cos \frac{(k-2)M\pi}{n} - \cos \frac{kM\pi}{n}}{\cos \frac{(k-2)M\pi}{n} - \cos \frac{(kM+2)\pi}{n}}\nonumber \\
  & < 1,
 \label{eq:ratio5}
 \end{align}
where the inequality follows for $\cos \frac{(kM+2)\pi}{n} <  \cos \frac{kM\pi}{n}$.
Finally, the other cases, where two or more rows change, can be proved comparing their determinant with their ancestor's with a similar reasoning.
 This completes the proof that  a systematic frame with the most evenly spaced systematic rows,
 or equivalently data samples in the corresponding codewords, is the best in the minimum MSE sense.

\end{proof}

\end{document}